\documentclass[letterpaper, 10 pt, conference]{ieeeconf} 
\IEEEoverridecommandlockouts                            
\usepackage{balance}
\usepackage{cite}
\usepackage{amsmath,amssymb,amsfonts}
\usepackage{algorithmic}
\usepackage{graphicx}
\usepackage{textcomp}
\usepackage{xcolor}
\usepackage{flushend}
\usepackage{booktabs} 
\usepackage[ruled, vlined, linesnumbered]{algorithm2e}
\def\BibTeX{{\rm B\kern-.05em{\sc i\kern-.025em b}\kern-.08em
    T\kern-.1667em\lower.7ex\hbox{E}\kern-.125emX}}
\usepackage[hidelinks]{hyperref}

\usepackage{amsthm}

\newtheorem*{remark}{Remark}
\newtheorem{theorem}{Theorem}
\newtheorem{lemma}{Lemma}

\hypersetup{
    urlcolor=grey,
}

\title{\LARGE \bf
Barrier Certified Safety Learning Control: \\When Sum-of-Square Programming Meets Reinforcement Learning
}

\author{Hejun Huang$^{1}$, Zhenglong Li$^{2}$, Dongkun Han$^{1*}$% <-this % stops a space
\thanks{$^{1}$Hejun Huang and Dongkun Han are with the Department of Mechanical and Automation Engineering, 
        The Chinese University of Hong Kong}
\thanks{$^{2}$Zhenglong Li is with the Department of Electrical and Electronic Engineering, The University of Hong Kong}
\thanks{*Email: {\tt\small dkhan@mae.cuhk.edu.hk}}
}

\begin{document}

\maketitle
\thispagestyle{empty}
\pagestyle{empty}

%%%%%%%%%%%%%%%%%%%%%%%%%%%%%%%%%%%%%%%%%%%%%%%%%%%%%%%%%%%%%%%%%%%%%%%%%%%%%%%%
\begin{abstract}
Safety guarantee is essential in many engineering implementations. Reinforcement learning provides a useful way to strengthen safety. However, reinforcement learning algorithms cannot completely guarantee safety over realistic operations. To address this issue, this work adopts control barrier functions over reinforcement learning, and proposes a compensated algorithm to completely maintain safety. Specifically, a sum-of-squares programming has been exploited to search for the optimal controller, and tune the learning hyperparameters simultaneously. Thus, the control actions are pledged to be always within the safe region. The effectiveness of proposed method is demonstrated via an inverted pendulum model. Compared to quadratic programming based reinforcement learning methods, our sum-of-squares programming based reinforcement learning has shown its superiority.
\end{abstract}

% provided convex optimization conditions,

%%%%%%%%%%%%%%%%%%%%%%%%%%%%%%%%%%%%%%%%%%%%%%%%%%%%%%%%%%%%%%%%%%%%%%%%%%%%%%%%
\section{Introduction}
% What RL
Reinforcement learning (RL) is one of the most popular methods for accomplishing the long-term objective of a Markov-decision process \cite{sutton2018reinforcement,duan2016benchmarking,mnih2015human,silver2016mastering,lillicrap2015continuous,9560769}. This learning method aims for an optimal policy that can maximize a long-term reward in a given environment recursively. For the case of RL policy gradient method \cite{peters2008reinforcement, sutton1999policy, silver2014deterministic}, this reward-related learning problem is traditionally addressed using gradient descent to obtain superior control policies. 

% Why safe RL
RL algorithms might generate desired results in some academic trials, e.g., robotic experiments \cite{kober2013reinforcement,kormushev2013reinforcement} and autonomous driving contexts \cite{levinson2011towards,kiran2021deep}. But for realistic situations, how to completely guarantee the safety based on these RL policies is still a problem. The intermediate policies from the RL agent's exploration and exploitation processes may sometimes ruin the environment or the agent itself, which is unacceptable for safety-critical systems. Moreover, disturbances around the system would confuse the agent over learning process, which make the control tasks more complex and challenging.

% Recent work
To cope with the issues above, various methods are proposed to specify dynamical safety during the learning process. Initially, the region of attraction (ROA) from control Lyapunov functions or Lyapunov-like function work \cite{wang2018permissive,han19tac} are used to certify the dynamical safety. Related and useful work of ROA estimation are introduced in \cite{chesi2011domain,han12tii,han2014tcas1,han2016estimating}. In 2017, \cite{berkenkamp2017safe} started to execute the RL's exploration and exploitation in ROA to maintain safety. Different from region constraint, in 2018 \cite{alshiekh2018safe,2204.00755} proposed a \textit{shield} framework to select safe actions repeatedly. Inspired by this work, \cite{cheng2019end,pmlr-v97-cheng19a} further combined the control barrier function (CBF) to synthesis the RL-based controller. Note that, the CBF is not new to identify safe regions in control theory, \cite{hsu2015control,thirugnanam2022safety} are representative work to combine CBF in biped robot and \cite{wang2018safe} explored its implementation of quadrotors, etc. Appropriate region constraint for RL is important for building safe RL. Hence, a reliable learning method is expected to construct a controller inside the safe region over the process.

In this paper, we will consider a state-space model in the polynomial form, including an unknown term that estimated by Gaussian process. Then, we will apply Deep Deterministic Policy Gradient (DDPG) algorithm to design a controller. The generated control action will impact system dynamics and further influence the sum-of-squares programming (SOSP) solution of control barrier functions. 

The main contributions of this work are: (1) Formulate the framework to embed SOSP-based CBF into DDPG algorithm; (2) compare the performance of quadratic programming (QP)-based and SOSP-based controller with DDPG algorithm; and (3) demonstrate the safety and efficiency of DDPG-SOSP algorithm via an inverted pendulum model.

This paper is organized as follows, we first present some preliminaries about reinforcement learning, Gaussian process and control barrier function in Section \ref{sec:preliminary}. Then, in Section \ref{sec:SOSP}, we formulate the steps to compute an SOSP-based controller and introduce a corresponding DDPG-SOSP algorithm. Numerical examples are given for these cases in Section \ref{sec:Experiment}, then we discuss the learning efficiency and safety performance of QP-based and SOSP-based strategies, before concluding this paper.

\section{Preliminary} \label{sec:preliminary}

A discrete-time nonlinear control-affine system with state $s_{t},s_{t+1}\in S$ and control action $a_t\in A$ is considered as:   
\begin{equation}
s_{t+1} = f(s_t) + g(s_t) a_t + d(s_t), \forall t\in[\underline{t},\overline{t}],
\label{eq:transition_dynamics_2}
\end{equation}
\noindent
where $s_t, s_{t+1}$ are finished within a bounded time $t$, while $\underline{t}$ and $\overline{t}$ are constants.

\textit{Notation:} From the infinite-horizon discounted Markov decision process (MDP) theory \cite{puterman2014markov}, let $r:\mathcal{S}\times \mathcal{A} \rightarrow\mathcal{R}$ denote the reward of the current state-action pair $(s_t,a_t)$ in (\ref{eq:transition_dynamics_2}), let $P(s_{t},s_{t+1},a_t)=p(s_{t+1}\vert s_t,a_t)$ denote the probability distribution to next state $s_{t+1}$, and let $\gamma\in(0,1)$ denote a discount factor. Thus, a MDP tuple $\tau =(s_t, s_{t+1}, a_t, r_t, P_t, \gamma)$ establishes at each state $t$. 

\subsection{Reinforcement Learning}

RL focuses on training an agent to map various situations to the most valuable actions in a long-term process. During the process, the agent will be inspired by a cumulative reward $r$ to find the best policy $\pi(a\vert s)$. Various novel RL algorithms are proposed to select optimal $\pi(a\vert s)$. For more details of selecting $\pi$, we kindly refer interested readers to \cite{sutton2018reinforcement}. 

Given MDP tuple $\tau$, the RL agent is selecting an optimal policy $\pi^*$ by maximizing the expected reward $J(\pi)$
\begin{equation}
J(\pi) = \mathbb{E}_\tau[\sum^{\infty}_{t=0}\gamma^t r(s_t,a_t)].
\label{eq:reward}
\end{equation}

\noindent
Furthermore, the action value function $Q_\pi$ can be generated from $\tau$ and satisfies the Bellman equation as follows,
\begin{equation}
\begin{aligned}
Q_\pi(s_t,a_t) &= \mathbb{E}_{s_t,a_t}[r_t+\gamma \mathbb{E}_{a_{t+1}}[Q_\pi (s_{t+1},a_{t+1})]]\\
&=\mathbb{E}_{s_{t+1},a_{t+1},\dots}[\sum^{\infty}_{l=0}\gamma^l r(s_{t+l},a_{t+l})].
\label{eq:saeq}
\end{aligned}
\end{equation}

In this paper, we will use DDPG algorithm, a typical actor-critic and off-policy method, to handle the concerned dynamics of (\ref{eq:transition_dynamics_2}). Let parameters $\theta^{\pi}$ and $\theta^Q$ denote the neural network of actor and critic, respectively. By interacting with the environment repeatedly, the actor will generate an actor policy $\pi:\mathcal{S}\rightarrow\mathcal{A}$ and be evaluated by the critic via $Q_\pi:\mathcal{S}\times\mathcal{A}\rightarrow\mathcal{R}$. These actions will be stored in replay buffer to update the policy gradient w.r.t $\theta^{\pi}$, and the loss function w.r.t $\theta^Q$, once if the replay buffer is fulfilled.

However, the state-action of DDPG algorithm generates without a complete safety guarantee, which is unacceptable in safety critical situations. Inspired by the work of \cite{cheng2019end}, we propose a programming assist method to maintain safety over the learning process in Section \ref{sec:SOSP}.

\subsection{Gaussian Process}	

Gaussian processes (GP) estimate the system and further predict dynamics based on the prior data. A GP is a stochastic process that construct a joint Gaussian distribution with concerned states $\{s_1,s_2,\dots\}\subset S$. We use GP to estimate the unknown term $d$ in (\ref{eq:transition_dynamics_2}) during the system implementation. Mixed with an independent Gaussian noise $(0,\sigma_n^2)$, the GP model's prior data $\{d_0,d_1,\dots ,d_k\}$ can be computed from $d_k = s_{k+1}-f-g\cdot a_k$ indirectly. Then, a high probability statement of the posterior output $\hat{d}:S\rightarrow \mathbb{R}$ generates as
\begin{equation}
    \begin{aligned}
    d(s)\sim \mathcal{N}(m_d(s),\sigma_d(s)),
    \end{aligned}
\end{equation}
\noindent
where the mean function $m_d(s)$ and the variance function $\sigma_d(s)$ can describe the posterior distribution as
\begin{equation}
    \begin{aligned}
    m_d(s)-k_\delta \sigma_d(s)\leq d(s)\leq m_d(s)+k_\delta \sigma_d(s),
    \end{aligned}
\end{equation}

\noindent
with probability greater or equal to $(1-\delta)^k$, $\delta\in(0,1)$ \cite{lederer2019uniform}, where $k_\delta$ is a parameter to design interval $[(1-\delta)^k,1]$. Therefore, by learning a finite number of data points, we can estimate the unknown term $d(s_*)$ of any query state $s_*\in \mathcal{S}$.

As the following lemma declaims in \cite{huang2021estimating,han2022sum}, there exists a polynomial mean function to approximate the term $d$ via GP, which can predict the output of state $s_*$ straightly here.
\begin{lemma}
\label{lem:polynomial_model}
Suppose we have access to $k$ measurements of $d(s)$ in (\ref{eq:transition_dynamics_2}) that corrupted with Gaussian noise $(0,\sigma_n^2)$. If the norm unknown term $d(s)$ bounded, the following GP model of $d(s)$ can be established with polynomial mean function $m_d(s_*)$ and covariance function $\sigma_d^2(s_*)$,
\begin{equation}
	\label{eqn:optimize_output_distribution_mean}
	\begin{aligned}
		m_d(s_*) &= \varphi(s_*)^{\mathrm{T}} w,\\
		\sigma_d^2(s_*) &= k(s_*, s_*)-k_*^{\mathrm{T}}(K+\sigma_n^2I)^{-1}k_*,
	\end{aligned}
\end{equation}	
\noindent within probability bounds $[(1-\delta)^k,1]$, where $\delta\in(0,1)$, $s_*$ is a query state, $\varphi(s_*)$ is a monomial vector, $w$ is a coefficient vector, $[K]_{(i,j)}=k(s_i, s_j)$ is a kernel Gramian matrix and $k_*=[k(s_1,s_*), k(s_2,s_*), \dots, k(s_k, s_*)]^{\mathrm{T}}$.
\end{lemma}

Lemma \ref{lem:polynomial_model} generates a polynomial expression of $d(s)$ within probabilistic range $[(1-\delta)^k,1]$. Meanwhile, the nonlinear term $f(s)$ in (\ref{eq:transition_dynamics_2}) can be approximated by Chebyshev interpolants $P_k(s)$ and bounded remainder $\xi(s)$ in certain domain as $f(s) = P_k(s)+\xi(s)$ \cite{trefethen2019approximation}, where $k$ is the degree of the polynomial $P_k(x)$ \cite{trefethen2019approximation}. Now we convert (\ref{eq:transition_dynamics_2}) as

\begin{equation}\label{eqn:learnsys}
\begin{aligned}
s_{t+1}=P_k(s_t)+g(s_t)a_t+d_\xi(s_t).
\end{aligned}
\end{equation}
\noindent
The polynomial (\ref{eqn:learnsys}) is equal to (\ref{eq:transition_dynamics_2}) with a new term $d_\xi(s_t)=d(s_t)+\xi(s_t)$. Consequently, we can obtain a polynomial system within a probability range by learning $d_\xi(s_t)$,
\begin{equation}\label{eqn:usesys}
\begin{aligned}
s_{t+1}=P_k(s_t)+g(s_t)a_t+m_{d_\xi}(s_t),
\end{aligned}
\end{equation}
\noindent
which will be used to synthesis the controller in Section \ref{sec:SOSP}.

\subsection{Control Barrier Function}

The super-level set of the control barrier function (CBF) $h:\mathcal{S}\rightarrow \mathcal{R}$ could validate an safety set $\mathcal{C}$ as
\begin{equation}
    \begin{aligned}
        \mathcal{C}=\{{s_t\in\mathcal{S}:h(s_t)\geq 0}\}.
    \end{aligned}\label{CBF}
\end{equation}
\noindent
Throughout this paper, we refer to $\mathcal{C}$ as a safe region and
\begin{equation}
    \begin{aligned}
        \mathcal{U}=\mathcal{S}-\mathcal{C} =\{s_t\in\mathcal{S}:h(s_t)<0\},
    \end{aligned}
\end{equation}
\noindent
as unsafe regions of the dynamical system (\ref{eqn:usesys}). Then, the state $s_t\in\mathcal{C}$ will not enter into $\mathcal{U}$ by satisfying the forward invariant constraint below,
\begin{equation}\label{eq:cbf-definition}
  \forall s_t\in \mathcal{C}:h(s_t)\geq 0,\,\Delta{h}(s_t,a_t) \geq 0,
\end{equation}
\noindent
where $\Delta{h(s_t,a_t)}=h(s_{t+1})-h(s_t)$. Before we demonstrate the computation of $h$ of the system (\ref{eqn:usesys}), the Positivestellensatz (P-satz) needs to be introduced first \cite{putinar1993positive}.

Let $\mathcal{P}$ be the set of polynomials and $\mathcal{P}^\text{SOS}$ be the set of sum of squares polynomials, e.g., $P(x)=\sum_{i=1}^{k}p_i^2(x), $ where $p_i(x)\in\mathcal{P}$ and $P(x)\in \mathcal{P}^{SOS}$.

\begin{lemma}\label{lem:psatz} For polynomials $\{a_i\}_{i=1}^m$, $\{b_j\}_{j=1}^n$ and $p$, define a set $\mathcal{B}=\{s\in\mathcal{S}:\{a_i(s)\}_{i=1}^m=0, \{b_j(s)\}_{j=1}^n\geq0\}$. Let $\mathcal{B}$ be compact. The condition $\forall x \in \mathcal{S}, p(s)\geq0$ holds if the following condition is satisfied,

\vspace{7pt}\centering{\hspace{32pt}
	$\left\{
	\begin{array}{l}
		\exists r_1,\dots, r_m \in \mathcal{P}, ~ s_1,\dots, s_n \in \mathcal{P}^{\text{SOS}}, \\
		p-\sum^{m}_{i=1}r_i a_i-\sum^{n}_{j=1}s_j b_j \in \mathcal{P}^{\text{SOS}}.
	\end{array} 
	\right. $}
\hfill$\square$
\end{lemma}

Lemma \ref{lem:psatz} points out that any strictly positive polynomial $p$ lies in the cone that generated by non-negative polynomials $\{b_j\}_{j=1}^n$ in the set of $\mathcal{B}$. Based on the definition of $\mathcal{C}$ and $\mathcal{U}$, P-satz will be adequately used in the safety verification. 

\begin{figure*}[ht] 
	\centering
	\includegraphics[width=0.95\linewidth]{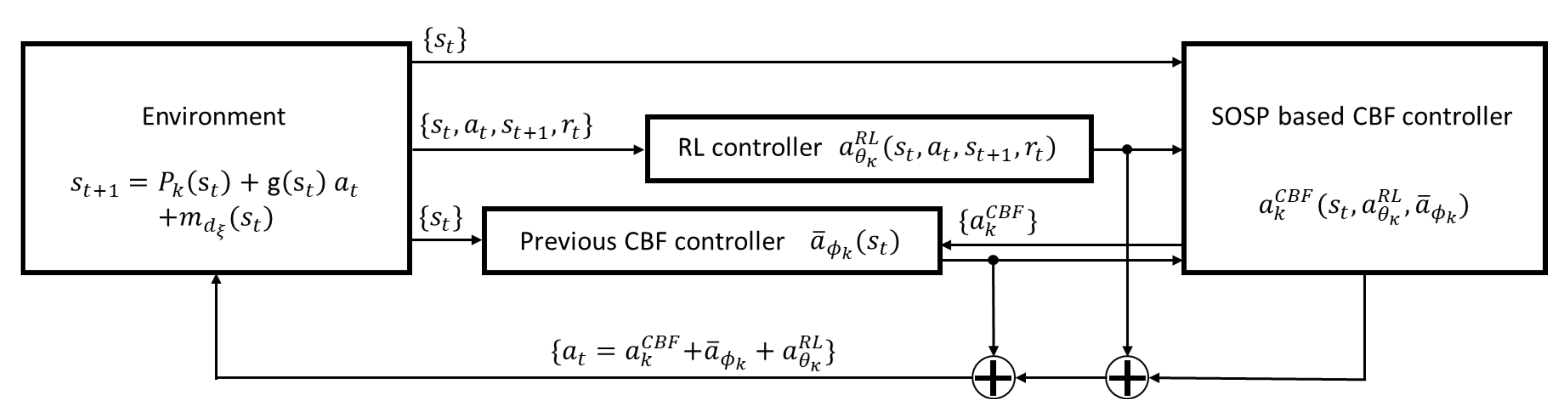}
	\caption{Workflow of the safe RL control with SOS program}
	\label{fig:architecture}
\end{figure*}

The sufficient condition to establish a CBF is the existence of a deterministic controller $a^{CBF}\colon\mathcal{S}\rightarrow\mathcal{A}$. Based on Lemma \ref{lem:psatz}, we can compute a polynomial controller $a^{CBF}$ through SOS program to avoid entering into unsafe regions $\mu_i(s)\in \mathcal{U}, i=1,2,\dots,n$ at each step $t$.

\begin{lemma}\label{lem:CBF—action}
Provided a safe region $\mathcal{C}$ certified by the polynomial barrier function $h$ and some unsafe regions $\mu_i\in\mathcal{U},i\in\mathbb{Z}^+$, in a polynomial system (\ref{eqn:usesys}), if there exists a controller $a^{CBF}$ that satisfies
\begin{equation}
\begin{aligned}\label{eq:CLF-CBF-discrete}
	& && \qquad \underset{s_t\in \mathcal{S},a^{CBF}\in\mathcal{A}; L(s_t),M_i(s_t)\in\mathcal{P}^{SOS};\;\;}{a^*\quad =\quad \arg\min \;{\Vert a^{CBF} \Vert_2}}\\
	&\text{s.t.} &&\;\quad\Delta{h}(s_t,a^{CBF}) - L(s_t)h(s_t)\in\mathcal{P}^{SOS},\\
	& && -\Delta{h}(s_t,a^{CBF})-M_i(s_t)\mu_i(s_t)\in \mathcal{P}^{\text{SOS}},
\end{aligned}
\end{equation}
\noindent
where $L(s_t)$ and $M_i(s_t)$ are SOS polynomials for $i=1,2,\dots,n$. Then, action $a^*$ is a minimum polynomial controller that regulates the system (\ref{eqn:usesys}) safely.
\end{lemma}

\begin{proof}
Let us suppose that there exists a CBF $h$ that certified a safe region $\mathcal{C}$. Then, a deterministic input $a_t=a^{CBF}(s_t)$ is needed to establish the necessary condition $\Delta{h(s_t,a^{CBF})}\geq 0$ in (\ref{eq:cbf-definition}) of the system (\ref{eqn:usesys}). 

Since $h$ is a polynomial function, if $a^{CBF}$ is also a polynomial, $\Delta h$ will maintain in the polynomial form. According to Lemma \ref{lem:psatz}, we can formulate this non-negative condition of $\Delta h$ as part of SOS constraint in the cone of $h$ as
\begin{equation}
    \begin{aligned}\label{eq:text1}
		&\Delta h(s_t,a^{CBF}) - L(s_t)h(s_t)\in\mathcal{P}^{SOS},
    \end{aligned}
\end{equation}
\noindent
where the auxiliary SOS polynomials $L(s_t)$ is used to project the related term into an non-negative compact domain $\mathcal{C}$ that certified by the superlevel set of $h(s_t)$. 

We continue to leverage the P-satz again to drive the dynamics of (\ref{eqn:usesys}) far from unsafe areas $\mu_i(s_t)$ as follows,
\begin{equation}
    \begin{aligned}\label{eq:text2}
        -\Delta h(s_t,a^{CBF}) - \mu_i(s_t)M_i(s_t)\in\mathcal{P}^{SOS},
    \end{aligned}
\end{equation}
\noindent
where $M_i(s_t)$ denote the corresponding SOS polynomials. Guided by the optimization (\ref{eq:CLF-CBF-discrete}), we can solve a minimal cost control to maintain system safety with (\ref{eq:text1}) and (\ref{eq:text2}), which completes the proof.
\end{proof}
This lemma maintains the safety of polynomial system (\ref{eqn:usesys}) and minimizes the current control $a^{CBF}$ with a SOS program. Multiple unsafe regions $\mu_i(s_t)$ are also considered in solving the target input $a^*$ in (\ref{eq:CLF-CBF-discrete}). Since Lemma 1 discussed the probabilistic relationship of (\ref{eqn:usesys}) and (\ref{eq:transition_dynamics_2}), the computed result of (\ref{eq:CLF-CBF-discrete}) of (\ref{eq:transition_dynamics_2}) can also be a feasible control to regulate the true dynamics safely.

\section{CBF Guiding Control with Reinforcement Learning}\label{sec:SOSP}	

In the first part of this section, we will further express the computation steps of (\ref{eq:CLF-CBF-discrete}). In the second part, we will go through the details about the SOSP-based CBF guiding control with DDPG algorithm.

\subsection{SOS Program Based CBF}
Different from \cite{cheng2019end}, our control $a^{CBF}$ is solved by SOS program with polynomial barrier functions, rather than the QP and linear barrier functions. In \cite{cheng2019end}, they initially constructed QP with monotonous linear barrier constraints to cope with RL. Obviously, a safe but more relaxed searching area will enhance the RL training performance. We propose a SOS program to compute a minimal control with polynomial barrier function based on Lemma \ref{lem:CBF—action}. 

Lemma \ref{lem:CBF—action} declared that we can compute a minimal control $a^{CBF}$ from an approximated polynomial system (\ref{eqn:usesys}) directly. However, it is hard to define the minimization of the convex function $a^{CBF}$ in (\ref{eq:CLF-CBF-discrete}). The optimal control can be searched by using the following \textbf{2} steps:

\textbf{Step 1:} Search an optimal auxiliary factor $L(s_t)$ by maximizing the scalar $\epsilon_1$,
\begin{equation}
	\label{bcalgo:step2}
	\begin{aligned}
		& && \underset{\epsilon_1\in\mathcal{R}^{+}, L(s_t)\in\mathcal{P}^{\text{SOS}}}{L^*(s_t) \quad = \quad \arg\max \quad \epsilon_1}\\
		&  \text{s.t.}
		&& \Delta h(s_t,a^{CBF}) - L(s_t) h(s_t) - \epsilon_1 \in \mathcal{P}^{\text{SOS}}\\
		& &&-\Delta{h}(s_t,a^{CBF})-M_i(s_t)\mu_i(s_t)\in \mathcal{P}^{\text{SOS}},
	\end{aligned}
\end{equation}

\noindent where $L(s_t)$ is an auxiliary factor to obtain a permissive constraint for the existing action $a^{CBF}$.

\textbf{Step 2:} Given $L(s_t)$ from last step, search an optimal control of $a^{CBF}$ with minimal control cost by minimizing the scalar $\epsilon_2$,
\begin{equation}
	\label{bcalgo:step3}
	\begin{aligned}
		& &&\underset{\epsilon_2\in\mathcal{R}^{+}}{a^*_t\quad = \quad \arg\min \quad \epsilon_2} \\
		& \text{s.t.} && \Delta{h}(s_t,a^{CBF}) - L(s_t) h(s_t) - \epsilon_2 \in \mathcal{P}^{\text{SOS}}\\
		& &&-\Delta{h}(s_t,a^{CBF})-M_i(s_t)\mu_i(s_t)\in \mathcal{P}^{\text{SOS}}.
	\end{aligned}
\end{equation}

\begin{remark}
The scalars $\epsilon_1$ and $\epsilon_2$ in (\ref{bcalgo:step2}) and (\ref{bcalgo:step3}) are used to limit the magnitude of the coefficients' optimization of $L(s_t)$ and $a_t$, respectively.
\end{remark}

The SOS programs above demonstrate the details of a target controller computation of the system (\ref{eqn:usesys}), and the solution of (\ref{bcalgo:step3}) regulates the dynamical safety via a control barrier function directly.

\subsection{SOS Program Based CBF with Reinforcement Learning}\label{ssec:SOS}

The workflow of the CBF guiding DDPG control algorithm is shown in Fig. \ref{fig:architecture}, which is akin to the original idea in \cite{cheng2019end} to push the agent's action into a safe training architecture. In Fig. \ref{fig:architecture}, we list these flowing elements into brackets and highlight the corresponding input of each factor computation. 

We illustrate the core idea of CBF guiding RL control as,
\begin{equation}\label{eqn:cbf_rl_control}
    \begin{aligned}
        a_t = &a^{RL}_{\theta_k} +\bar{a}_k+a_t^{CBF},
    \end{aligned}
\end{equation}
\noindent
where the final action $a_t$ consists of a RL-based controller $a^{RL}_{\theta_k}$, a previous deployed CBF controller $\bar{a}_{\phi}$ and a SOS program based CBF controller $a_t^{CBF}$. A clear distinction of the subscript $t$ and $k$ is stated here: $t$ denotes the step number of each policy iteration and $k$ denotes the policy iteration number over the learning process.

More specifically, the first term in (\ref{eqn:cbf_rl_control}) is an action generated from a standard DDPG algorithm with parameter $\theta$. The second term $\bar{a}_{k}=\sum_{i=0}^{k-1}a_i^{CBF}$ demotes a global consideration of previous CBF controllers. Since it is impractical to compute the exact value of each CBF action $a_i^{CBF}$ at each step. Supported by the work of \cite{cheng2019end}, we approximate this term as $ \bar{a}_{\phi_k} \approx \bar{a}_k = \sum^{k-1}_{i=0}a_i^{CBF}$, where $\bar{a}_{\phi_k}$ denotes an output from the multilayer perceptron (MLP) with hyperparameters $\phi$. Then, we fit the MLP and update $\phi_{k}$ with $\sum^{k-1}_{i=0}a_i^{CBF}(s,a_{\theta_0}^{RL},\dots,a_{\theta_{i-1}}^{RL})$ at each episode.

The third term $a^{CBF}_t$ in (\ref{eqn:cbf_rl_control}) is a real-time value based on the deterministic value $s_t$, $a^{RL}_{\theta_k}(s_t)$ and $\bar{a}_{\phi_{k}}(s_t)$. Although there exists an unknown term $d$ in the system (\ref{eq:transition_dynamics_2}), we can still solve the dynamical safety in the approximated polynomial system (\ref{eqn:usesys}). So, the optimal $a_t$ in (\ref{eqn:cbf_rl_control}) can be computed by the SOS program below with deterministic $a^{RL}_{\theta_k}$ and $\bar{a}_{\phi_k}$
\begin{equation}
\begin{aligned}\label{eq:sosclb}\small
	& && \qquad \underset{s_t\in \mathcal{S},a_k\in\mathcal{A}; L(s_t),M_i(s_t)\in\mathcal{P}^{SOS};\;\;}{a^*\quad =\quad \arg\min \;{\Vert a_k \Vert_2}}\\
	&\text{s.t.} &&\;\quad\Delta{h}(s_t,a_k) - L(s_t)h(s_t)\in\mathcal{P}^{SOS},\\
	& && -\Delta{h}(s_t,a_k)-M_i(s_t)\mu_i(s_t)\in \mathcal{P}^{\text{SOS}}.
\end{aligned}
\end{equation}
\noindent
Thus, we can establish a controller by satisfying the solution of (\ref{eq:sosclb}) over the learning process.

\begin{theorem}
Given a barrier function $h$ in (\ref{CBF}), a partially unknown dynamical system (\ref{eq:transition_dynamics_2}) and a corresponding approximated system (\ref{eqn:usesys}), while (\ref{eqn:usesys}) is a stochastic statement of (\ref{eq:transition_dynamics_2}) within the probabilistic range $[(1-\delta)^k,1]$, suppose there exists an action $a_t$ satisfying (\ref{eq:sosclb}) in the following form:
\begin{equation}\label{eqn:the1}
    \begin{aligned}
    a_t(s_t) = &a^{RL}_{\theta_k}(s_t,a_t,s_{t+1},r_t)+\bar{a}_{\phi_k}(s_t,\sum_{i=0}^{k-1}(s_t,a^{CBF}_i))\\
    &+a^{CBF}_{t}(s_t,a^{RL}_{\theta_k},\bar{a}_{\phi_k}).
    \end{aligned}
\end{equation}
\noindent
Then, the controller (\ref{eqn:the1}) guarantees the system (\ref{eq:transition_dynamics_2}) inside the safe region within the range of probability $[{(1-\delta)}^n,1]$.
\end{theorem}
 
\begin{proof}
Regarding the system (\ref{eqn:usesys}), this result follows directly from Lemma \ref{lem:CBF—action} by solving the SOS program (\ref{eq:sosclb}). The only different part of SOS program (\ref{eq:sosclb}) from the SOS program (\ref{eq:CLF-CBF-discrete}) in Lemma \ref{lem:CBF—action} is the action $a_k$ here contains additional deterministic value $a_{\theta_k}^{RL}(s_t)$ and $\bar{a}_{\phi_{k}}(s_t)$. 

Since the output of (\ref{eq:sosclb}) can regulate the dynamics of the model (\ref{eqn:usesys}) in a safe region, while the approximated model (\ref{eqn:usesys}) is a stochastic statement of the original system (\ref{eq:transition_dynamics_2}) with the probability greater or equal to $(1-\delta)^k$. Then, the solution of (\ref{eq:sosclb}) can be regarded as a safe control to drive the system (\ref{eq:transition_dynamics_2}) far from dangerous, which ends the proof. 
\end{proof}

We display an overview of the whole steps in the combination of the aforementioned factors over the learning process. The detailed workflow is outlined in Algorithm \ref{alg:rl-cbf}

\begin{algorithm}[ht]
	\caption{DDPG-SOSP}\label{alg:rl-cbf}
	\KwIn{Origin system (\ref{eq:transition_dynamics_2}); barrier function $h$.}
	\KwOut{RL optimal policy $\pi$.}
	Preprocess (\ref{eq:transition_dynamics_2}) into (\ref{eqn:usesys}).\\		
	Create the tuple $\hat{D}$ to store $\{s_t,a_t,s_{t+1},r_t\}$ in $\hat{D}$.\\
	\For{$k\in\{1,2,\dots ,k\}$}{	
		Execute SOSP (\ref{eq:sosclb}) to obtain the real-time CBF controller $a^{CBF}_{k}$ and store it to train previous CBF controller $\bar{a}_{\theta_k}^{CBF}$.\\
		Actuate the agent with $a_t$ in (\ref{eqn:the1}).\\
		Construct $\hat{D}$ and update $\theta_k$ and $\phi_k$ directly.\\
	}	
	\Return $\pi$.\\
\end{algorithm}

\section{Numerical Example}\label{sec:Experiment}

A swing-up pendulum model from OpenAI Gym environment (\textit{Pendulum-v1}) is used to verify the above theorems,
\begin{equation}
    \begin{aligned}\label{eqn:model}
    ml^2\Ddot{\theta} = mgl \sin(\theta)+a.
    \end{aligned}    
\end{equation}
\noindent

Let $s_1=\theta$, $s_2=\dot{\theta}$. The dynamics of (\ref{eqn:model}) are defined as 
\begin{eqnarray}\label{demo:2D}
	\begin{aligned}
		\begin{bmatrix} \dot{s}_1 \\ \dot{s}_2 \end{bmatrix} = 
		\begin{bmatrix}
			s_{2}\\	\;-\frac{g}{l}\sin{(s_1)} +\frac{a}{ml^2}+d\;
		\end{bmatrix},
	\end{aligned}   
\end{eqnarray}
\noindent
where $d$ denotes unknown dynamics that generate by inaccurate parameters $\bar{m}=1.4$, $\bar{l}=1.4$, as \cite{cheng2019end} introduced. When we try to construct a polynomial system (\ref{eqn:usesys}), the $\sin(s_1)$ in (\ref{demo:2D}) will be approximated by Chebyshev interpolants within $s_1\in[-3,3]$. Then, with SOSOPT Toolbox \cite{seiler2013sosopt}, Mosek solver \cite{mosek2010mosek} and Tensorflow, we implement Algorithm 1 in the pendulum model. 

The related model parameters are given in Table 1.
\begin{center}
	\rm{TABLE 1}: Swing-up Pendulum Model Parameters\label{tab:test}\\
	\setlength{\tabcolsep}{1.5mm}{
		\begin{tabular}{cccc} \toprule
			Model Parameter &Symbol &Value &Units\\
			\midrule
			Pendulum mass & $m $ & $1.0 $ & $\text{kg}$ \\ 
			Gravity & $g$ & $10.0 $ & $\text{m}/\text{s}^2$ \\ 
			Pendulum length & $l$ & $1.0 $ & $\text{m}$ \\ 
			Input torque & $a$ & $[-15.0,15.0] $ & $\text{N}$ \\ 
			Pendulum Angle & $\theta $ & $[-1.00,1.00]$ & $\text{rad}$ \\ 
			Angle velocity & $\dot{\theta} $ & $[-60.0,60.0] $ & $\text{rad}/\text{s}$ \\ 
			Angle acceleration & $\Ddot{\theta}$ & --- & $\text{rad}/\text{s}^2$ \\ 
			\bottomrule
	\end{tabular}}
\end{center}

We want to maintain the pendulum angle $\theta$ always in a safe range $[-1.0,1.0]$ during the learning process. Three RL algorithms are selected and compared in this example, including the DDPG-ONLY algorithm \cite{lillicrap2015continuous}, CBF-based DDPG-QP algorithm \cite{cheng2019end} and CBF-based DDPG-SOSP algorithm (our work). The codes of CBF-based DDPG-SOSP can be found at the url: \href{https://github.com/Wogwan/CCTA2022\_SafeRL}{https://github.com/Wogwan/CCTA2022\_SafeRL}. 

All of these algorithms are trained in $150$ episodes and each episode contains $200$ steps. Each step is around $0.05s$. And the reward function $r=\theta^2+0.1\dot{\theta}^2+0.001a^2$ is defined such that the RL agent are expected to keep the pendulum upright with minimal $\theta$, $\dot{\theta}$ and $\Ddot{\theta}$. 

\begin{figure}[t] 
	\centering
	\includegraphics[width=0.84\linewidth]{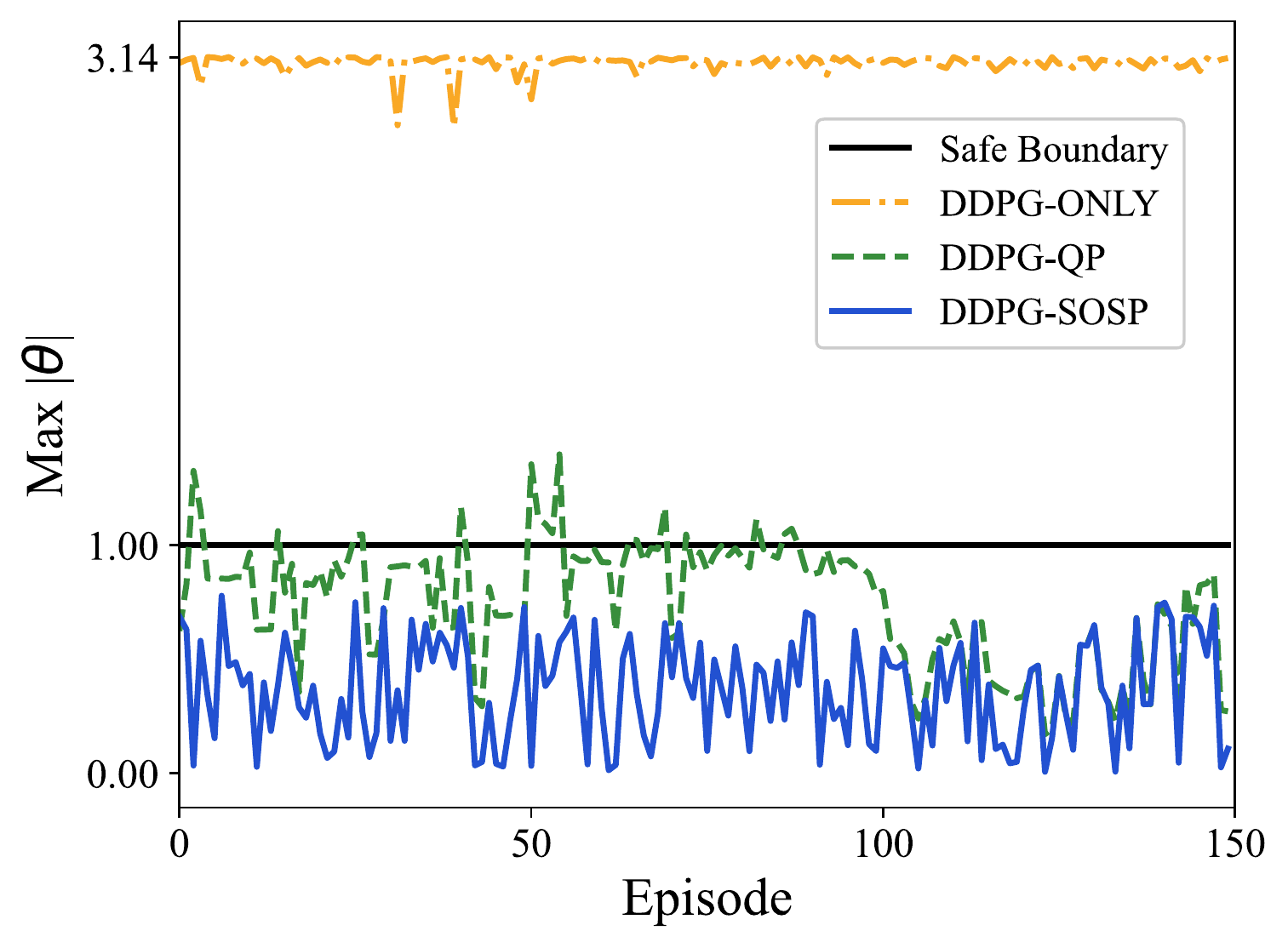}
	\caption{Comparison of the maximum absolute $\theta$ of different algorithms. The dotted solid line, the dashed line, and the solid line denote the performance of DDPG-ONLY, DDPG-QP and DDPG-SOSP, respectively The straight line denotes the safe boundary $\vert s_1\vert = 1$. }
	\label{fig:maxangle}
\end{figure}

\begin{figure}[htbp] 
	\centering
	\includegraphics[width=0.84\linewidth]{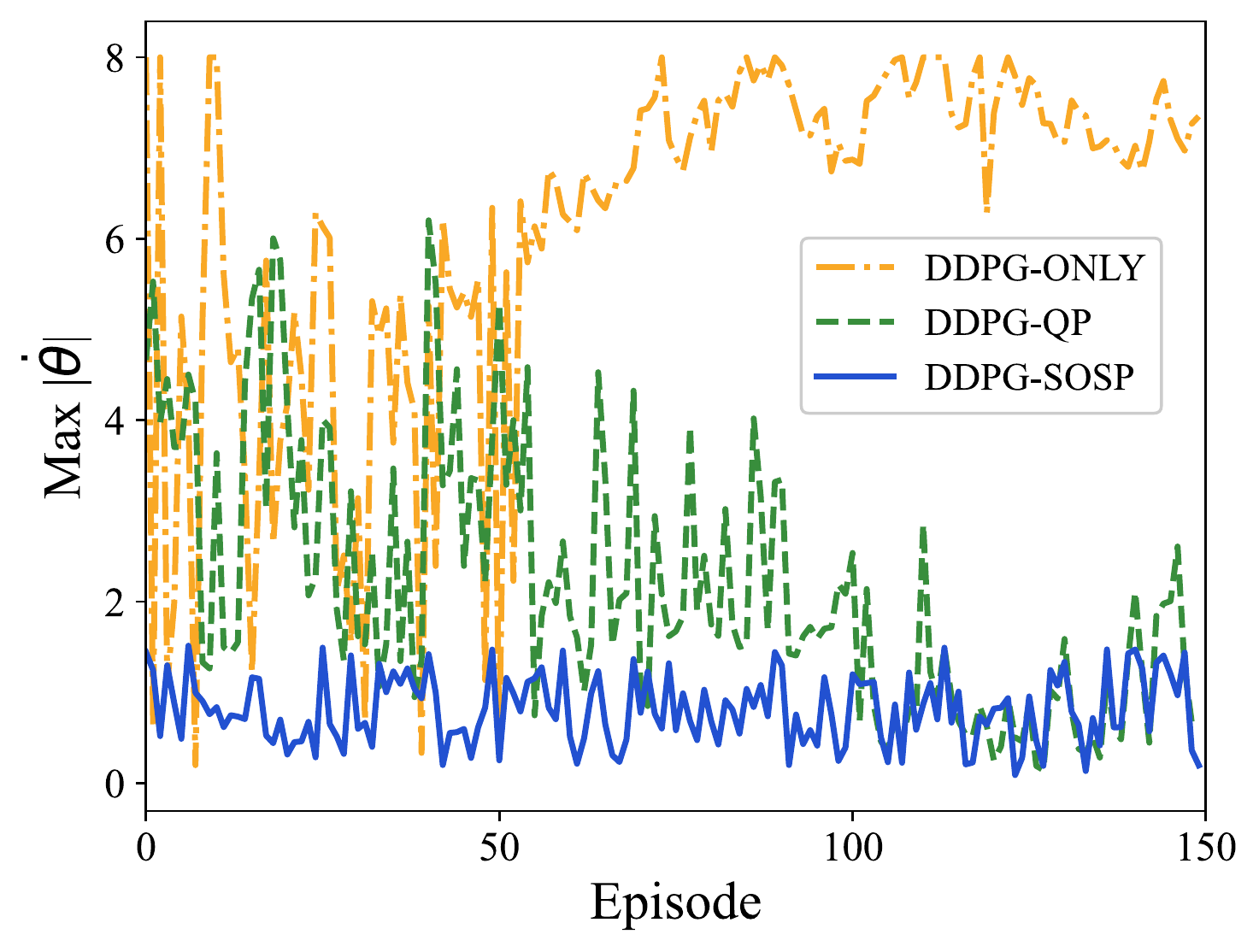}
	\caption{Comparison of the maximum absolute $\dot{\theta}$ of different algorithms. }
	\label{fig:maxanglespeed}
\end{figure}

\begin{figure}[htbp] 
	\centering
	\includegraphics[width=0.84\linewidth]{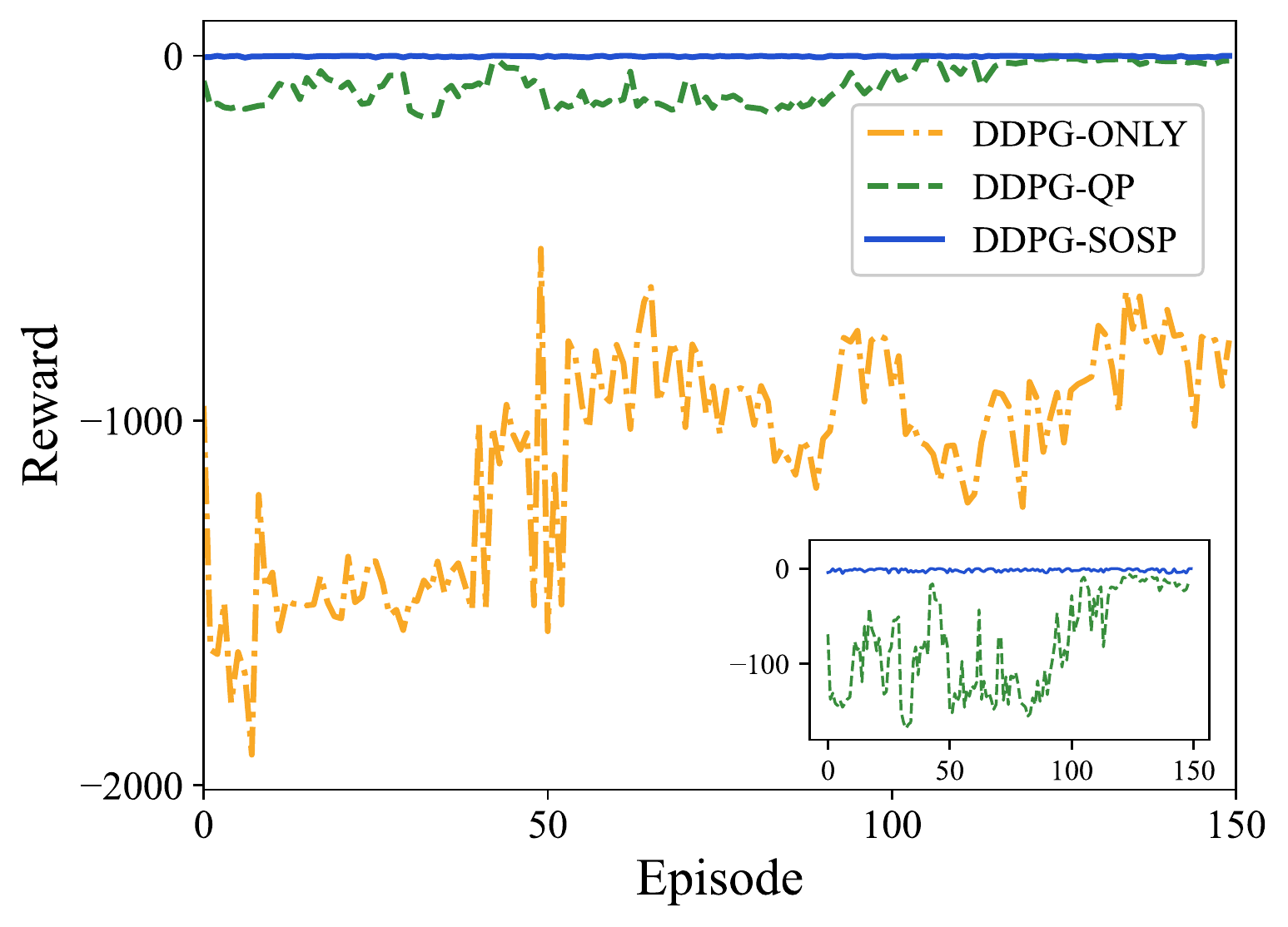}
	\caption{Comparison of the accumulated reward of different algorithms. }
	\label{fig:reward}
\end{figure}

\begin{figure}[b]
	\centering 
	\begin{minipage}{0.235\textwidth}
		\centering \label{fig:qp_1}
		\includegraphics[width=\textwidth]{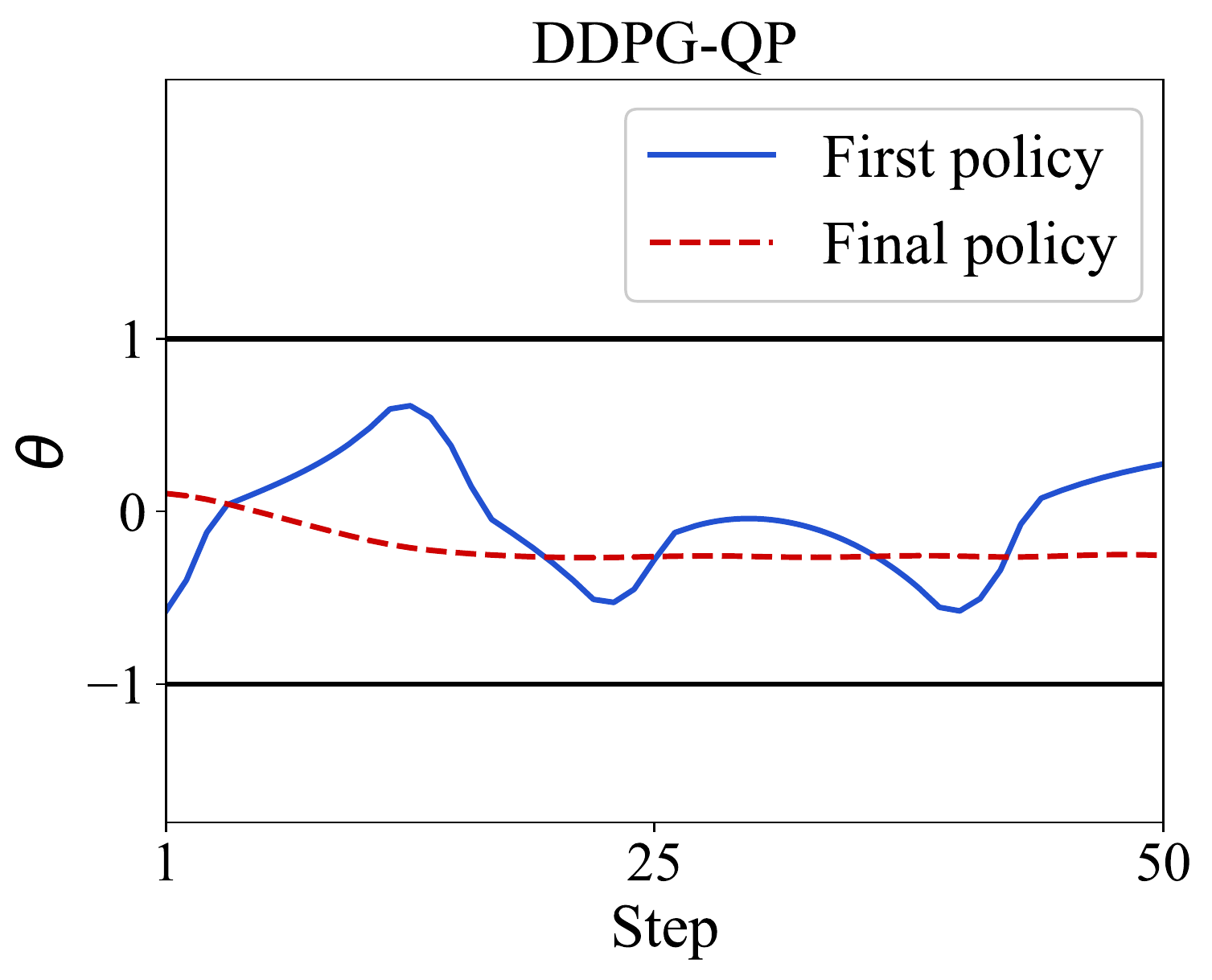}
		\centerline{(a)}
	\end{minipage}
	\begin{minipage}{0.235\textwidth}
		\centering \label{fig:sosp_1}
		\includegraphics[width=\textwidth]{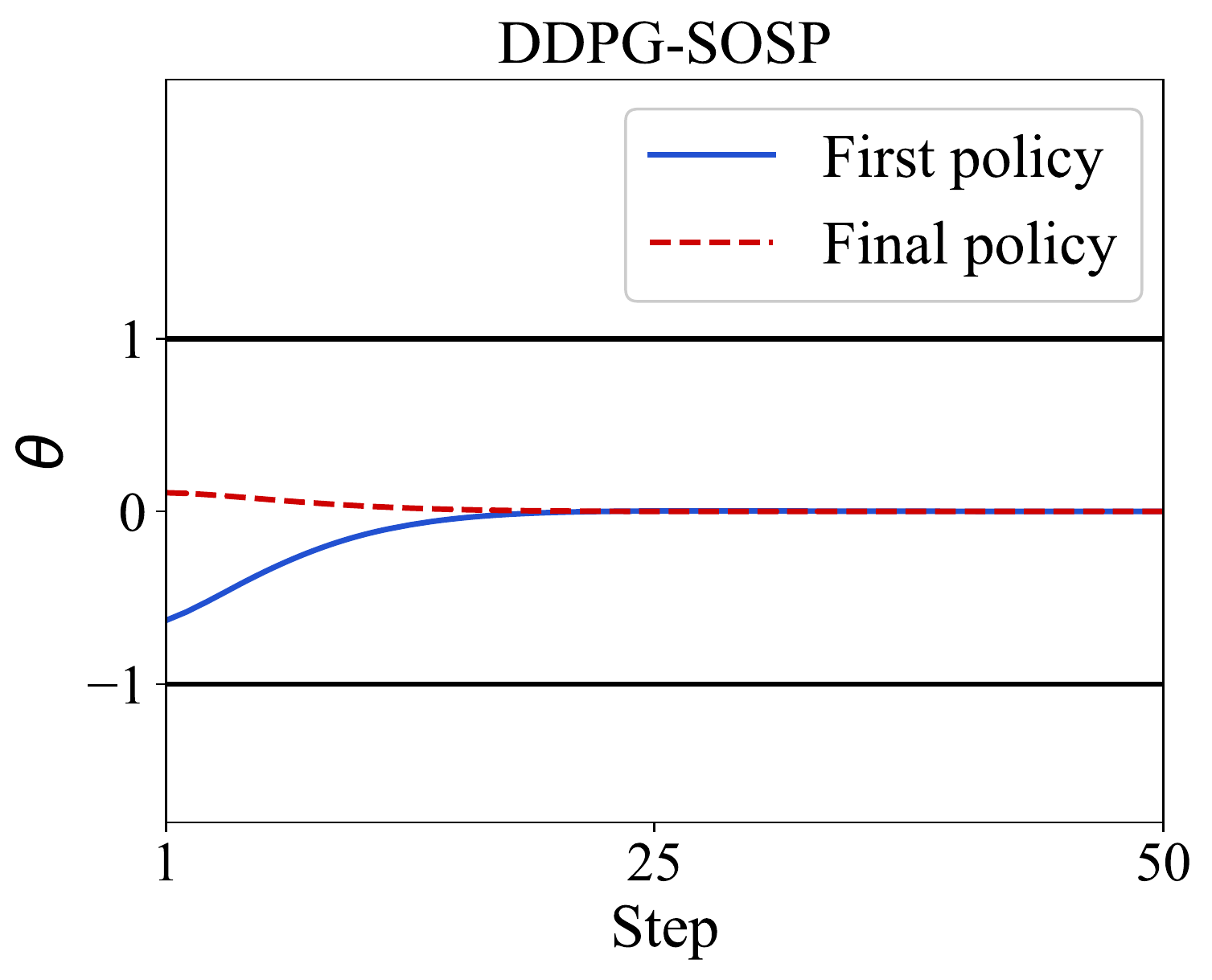}
		\centerline{(b)}
	\end{minipage}
	\caption{The comparison of the $1^{st}$ and the $150^{th}$ policy performance of the angle control task between $(a)$ DDPG-QP and $(b)$ DDPG-SOS.}
	\label{fig:fig_Comparison}
\end{figure}

Fig. \ref{fig:maxangle} compares the collected maximum $\vert\theta\vert$ at each episode by using different algorithms. As a baseline algorithm, DDPG-ONLY explored every state without any safety considerations, while DDPG-QP sometimes violated the safe state and DDPG-SOS completely kept in safe regions. 

The historical maximum value of $\vert \dot{\theta}\vert$ of these algorithms is compared and shown in Fig. \ref{fig:maxanglespeed}. It is found that DDPG-SOS is able to maintain the pendulum into a lower $\dot{\theta}$ easily than others over episodes. 

In Fig. \ref{fig:reward}, from the reward curve of these algorithms, it is easy to observe the efficiency of different algorithms: DDPG-SOS obtains comparatively lower average reward and thus a better performance.

Fig. \ref{fig:fig_Comparison} shows the first ($2^{nd}$ episode) and the final ($150^{th}$ episode) policy guided control performance between DDPG-QP and DDPG-SOSP algorithm. We observed $50$ steps of the pendulum angle to highlight the superiority of these two algorithms. Both in Fig. \ref{fig:fig_Comparison}(a) and (b), the final policy can reach a smoother control output with less time. Although DDPG-SOSP takes a quicker and smoother action to stabilize the pendulum than DDPG-QP, but the time consuming will increase due to the dynamics regression and optimization computation under uncertainty. Accordingly, the result of DDPG-SOSP algorithm of model (\ref{eqn:model}) is not only obviously stabler than DDPG-QP, but also maintaining the RL algorithm action's safety strictly.

\section{Conclusion}

In this paper, we consider a novel approach to guarantee the safety of reinforcement learning (RL) with sum-of-squares programming (SOSP). The objective is to find a safe RL method toward a partial unknown nonlinear system such that the RL agent's action is always safe. One of the typical RL algorithms, Deep Deterministic Policy Gradient (DDPG) algorithm, cooperates with control barrier function (CBF) in our work, where CBF is widely used to guarantee the dynamical safety in control theories. Therefore, we leverage the SOSP to compute an optimal controller based on CBF, and propose an algorithm to creatively embed this computation over DDPG. Our numerical example shows that SOSP in the inverted pendulum model obtains a better control performance than the quadratic program work. It also shows that the relaxations of SOSP can enable the agent to generate safe actions more permissively.

For the future work, we will investigate how to incorporate SOSP with other types of RL including value-based RL and policy-based RL. Meanwhile, we check the possibility to solve a SOSP-based RL that works in higher dimensional cases. 

% \section{Acknowledgements}
% The authors gratefully acknowledge the reviewers' helpful comments.

\balance

\end{document}